\documentclass[journal]{IEEEtran}
\usepackage{amsmath}
\usepackage{amssymb}
\usepackage{amsfonts}
\usepackage{graphicx}
\usepackage{epsfig}
\usepackage{subfigure}
\usepackage{psfrag}
\usepackage{color}
\usepackage{overpic}
\begin{document}

\title{Investigation on Multiuser Diversity in Spectrum Sharing Based Cognitive Radio Networks}

\author{Rui Zhang,~\IEEEmembership{Member,~IEEE}, and Ying-Chang
Liang,~\IEEEmembership{Senior Member,~IEEE}
\thanks{The authors are with the Institute for Infocomm Research, A*STAR, Singapore (Email: \{rzhang,
ycliang\}@i2r.a-star.edu.sg).}\vspace{-0.2in}}

\maketitle

\begin{abstract}
A new form of multiuser diversity, named \emph{multiuser
interference diversity}, is investigated for opportunistic
communications in cognitive radio (CR) networks by exploiting the
mutual interference between the CR and the existing primary radio
(PR) links. The multiuser diversity gain and ergodic throughput are
analyzed for different types of CR networks and compared against
those in the conventional networks without the PR link.
\end{abstract}


\IEEEpeerreviewmaketitle

\newtheorem{claim}{Claim}
\newtheorem{guess}{Conjecture}
\newtheorem{definition}{Definition}
\newtheorem{fact}{Fact}
\newtheorem{assumption}{Assumption}
\newtheorem{theorem}{\underline{Theorem}}[section]
\newtheorem{lemma}{\underline{Lemma}}[section]
\newtheorem{ctheorem}{Corrected Theorem}
\newtheorem{corollary}{\underline{Corollary}}[section]
\newtheorem{proposition}{Proposition}
\newtheorem{example}{\underline{Example}}[section]
\newtheorem{remark}{\underline{Remark}}[section]
\newtheorem{problem}{\underline{Problem}}[section]
\newtheorem{algorithm}{\underline{Algorithm}}[section]
\newcommand{\mv}[1]{\mbox{\boldmath{$ #1 $}}}

\section{Introduction}
\PARstart{C}ognitive radio (CR) is a promising technology for
efficient spectrum utilization in the future wireless communication
systems. The main design objective for a CR network is to maximize
its throughput while providing sufficient protection to the existing
primary radio (PR) network. There are two basic operation models for
CRs: {\it opportunistic spectrum access} (OSA), where the CR is
allowed to transmit over a frequency band only when all the PR
transmissions are off, and {\it spectrum sharing} (SS), where the CR
can transmit concurrently with PRs, provided that it knows how to
control the resultant interference powers at PRs below a tolerable
threshold \cite{Zhang09}. On the other hand, ``multiuser diversity
(MD)'' \cite{Knopp95,Tse02}, as a fundamental property of wireless
networks, has been widely applied for opportunistic communications
in wireless systems. The conventional form of MD is usually
exploited in a wireless system with multiple independent-fading
communication links by selecting one link with the best
instantaneous channel condition to transmit at one time, also known
as {\it dynamic time-division multiple-access} (D-TDMA).

In this letter, we consider SS-based fading CR networks and
investigate for them a new form of MD, named {\it multiuser
interference diversity} (MID), which is due to the CR and PR mutual
interference and thus different from the conventional MD. More
specifically, ``transmit MID (T-MID)'' is due to the fact that the
maximum transmit powers of different CR transmitters vary with their
independent-fading channels to a PR receiver under a given
interference power constraint, while ``receive MID (R-MID)'' is due
to the fact that the PR interference powers at different CR
receivers vary with their independent-fading channels from a PR
transmitter. This letter studies the MID for various types of CR
networks with D-TDMA, including the multiple-access channel (MAC),
broadcast channel (BC), and parallel-access channel (PAC), and
analyzes the achievable {\it multiuser diversity gain} (MDG) as a
function of the number of CRs, $K$. It is shown that due to the
newly discovered MID, the MDG (as defined in this letter) for each
CR network under consideration is no smaller than that in the
conventional network without the PR link, which holds for arbitrary
fading channel distributions and values of $K$. On the other hand,
it is also shown that the MID results in the same asymptotic growth
order over $K$ for the CR network ergodic throughput as that due to
the MD for the conventional network, as $K$ goes to infinity.

It is noted that MD for the fading CR networks has been recently
studied in, e.g., \cite{Ban09} and \cite{Tajer}. In \cite{Ban09}, MD
was analyzed for the fading CR MAC with a given $K$ at the
asymptotic regime where the ratio between the CR transmit power
constraint and interference power constraint goes to infinity. In
\cite{Tajer}, MD was investigated for the fading CR PAC with the
constant CR transmit power at the asymptotic region with
$K\rightarrow \infty$. In this letter, we study MD for different
types of CR networks with arbitrary CR transmit and interference
power constraints, and arbitrary number of CR users.

\section{System Model}\label{sec:system model}

Consider a SS-based CR network coexisting with a PR network. For the
propose of exposition, only one active PR link consisting of a PR
transmitter (PR-Tx) and a PR receiver (PR-Rx) is considered. All
terminals in the network are assumed to be each equipped with a
single antenna. We consider the {\it block-fading} (BF) model for
all the channels involved and coherent communications; thus, only
the fading channel power gains (amplitude squares) are of interest.
It is assumed that the additive noises at all PR and CR receive
terminals are independent circular symmetric complex Gaussian (CSCG)
random variables each having zero mean and unit variance. We
consider three types of CR networks described as follows.
\begin{itemize}
\item {\it Cognitive MAC (C-MAC)}, where $K$ CRs, denoted by CR$_1$, CR$_2,\ldots$, CR$_K$, transmit
independent messages to a CR base station (CR-BS). Denote $\mv{h}_k$
as the power gain of the fading channel from CR$_k$ to CR-BS,
$k=1,\ldots,K$.  Similarly, $\mv{g}_k$ is defined for the fading
channel from CR$_k$ to PR-Rx, $\mv{f}$ is for that from PR-Tx to
PR-Rx, and $\mv{e}$ is for that from PR-Tx to CR-BS.

\item {\it Cognitive BC (C-BC)}, where CR-BS sends independent messages to $K$ CRs.
Channel reciprocity is assumed between the C-MAC and C-BC; thus, the
fading channel power gain from CR-BS to CR$_k$ in the C-BC is the
same as $\mv{h}_k$ in the C-MAC. In addition, $\mv{g}$ is defined as
the fading channel power gain from CR-BS to PR-Rx, and $\mv{e}_k$ is
defined for the channel from PR-Tx to CR$_k$.

\item {\it Cognitive PAC (C-PAC)}, where $K$ distributed CR transmitters, denoted by CR-Tx$_1$, $\ldots$,
CR-Tx$_K$, transmit independent messages to the corresponding
receivers, denoted by CR-Rx$_1$, $\ldots$, CR-Rx$_K$, respectively.
Similarly like the C-MAC and C-BC, we denote $\mv{h}_k$ as the power
gain of the fading channel from CR-Tx$_k$ to CR-Rx$_k$,
$k=1,\ldots,K$; $\mv{g}_k$ and $\mv{e}_k$ are defined for the fading
channels from CR-Tx$_k$ to PR-Rx and from PR-Tx to CR-Rx$_k$,
respectively.
\end{itemize}

Let $J$ be the {\it peak} (with respect to fading states) transmit
power constraint at CR-BS in the C-BC, and $P_k$ be that at CR$_k$
in the C-MAC or at CR-Tx$_k$ in the C-PAC. It is assumed that PR-Tx
transmits with a constant power $Q$; and each CR transmit terminal
protects the PR by applying the {\it peak} interference power
constraint at PR-Rx, denoted by $\Gamma$. Let $p_k$ be the transmit
power for CR $k$ in the C-MAC or C-PAC, and $p$ be that for CR-BS in
the C-BC. Let $g_k$ be one realization of $\mv{g}_k$ for a
particular fading state (similar notations apply for the other
channels). Combining both the transmit and interference power
constraints, we obtain
\begin{align}\label{eq:power}
p_k \leq\min(P_k,\Gamma/g_k), \forall k, ~ p \leq \min(J,\Gamma/g).
\end{align}
The maximum achievable receiver signal-to-noise ratio (SNR) of CR
$k$ can then be expressed as
\begin{align}\label{eq:SNR}
\gamma_k^{\rm (MAC)}=\frac{h_kp_k}{1+Q e}, \gamma_k^{\rm
(BC)}=\frac{h_kp}{1+Qe_k}, \gamma_k^{\rm
(PAC)}=\frac{h_kp_k}{1+Qe_k}
\end{align}
for the C-MAC, C-BC, and C-PAC, respectively. Note that the noise at
each CR receive terminal includes both the additive Gaussian noise
and the interference from PR-Tx. We are then interested in
maximizing the long-term system throughput in each CR network by
adopting the following D-TDMA rule: At one particular fading state,
CR $k$ is selected for transmission if it has the largest achievable
receiver SNR among all the CRs. Let $k^*$ denote the selected user
at this fading state. From (\ref{eq:power}) and (\ref{eq:SNR}), it
then follows that
\begin{align}
k^*_{\rm MAC}=&\arg\max_{k\in\{1,\ldots,K\}} \frac{h_k\min(P_k,\Gamma/g_k)}{1+Qe} \label{eq:optimal k MAC}\\
k^*_{\rm BC}=&\arg\max_{k\in\{1,\ldots,K\}} \frac{h_k\min(J,\Gamma/g)}{1+Qe_k}  \label{eq:optimal k BC} \\
k^*_{\rm PAC}=&\arg\max_{k\in\{1,\ldots,K\}}
\frac{h_k\min(P_k,\Gamma/g_k)}{1+Qe_k}  \label{eq:optimal k PAC}
\end{align}
for the C-MAC, C-BC, and C-PAC, respectively. By substituting $k^*$
for each CR network into (\ref{eq:SNR}), the corresponding maximum
receiver SNR over CR users is obtained as $\gamma_{\rm MAC}(K)$,
$\gamma_{\rm BC}(K)$, or $\gamma_{\rm PAC}(K)$, each as a function
of $K$.

\section{Multiuser Interference Diversity}\label{sec:MU diversity}

In this section, we study the MD for different CR networks under a
set of ``symmetric'' assumptions, where all $\mv{h}_k$'s are assumed
to have the same distribution, so are $\mv{g}_k$'s and $\mv{e}_k$'s;
and $P_k=P, \forall k$. We then define $\bar{\gamma}_{\rm
MAC}(K)\triangleq\mathbb{E}[\gamma_{\rm
MAC}(K)]/\mathbb{E}[\gamma_{\rm MAC}(1)]$ as the {\it MDG} for the
C-MAC; $\bar{\gamma}_{\rm BC}(K)$ and $\bar{\gamma}_{\rm PAC}(K)$
are similarly defined for the C-BC and C-PAC, respectively. From
(\ref{eq:optimal k MAC}), (\ref{eq:optimal k BC}), and
(\ref{eq:optimal k PAC}), it follows that
\begin{align}
\bar{\gamma}_{\rm MAC}(K)&=\kappa_{\rm
MAC}\mathbb{E}[\max_k h_k\min(P,\Gamma/g_k)] \label{eq:MD gain MAC}\\
\bar{\gamma}_{\rm BC}(K)&=\kappa_{\rm
BC}\mathbb{E}[\max_k h_k/(1+Qe_k)] \label{eq:MD gain BC} \\
\bar{\gamma}_{\rm PAC}(K)&=\kappa_{\rm PAC}\mathbb{E}[\max_k
h_k\min(P,\Gamma/g_k)/(1+Qe_k)] \label{eq:MD gain PAC}
\end{align}
where
\begin{align}
\kappa_{\rm MAC} &
=1/(\mathbb{E}[h_k]\mathbb{E}[\min(P,\Gamma/g_k)]) \\
\kappa_{\rm BC}&=1/(\mathbb{E}[h_k]\mathbb{E}[1/(1+Qe_k)])
\\ \kappa_{\rm
PAC}&=1/(\mathbb{E}[h_k]\mathbb{E}[\min(P,\Gamma/g_k)]\mathbb{E}[1/(1+Qe_k)])
\label{eq:kappa PAC}
\end{align}
are constants, which can be shown independent of $k$ due to the
symmetric assumptions. In order to fairly compare the MDG in each CR
network with that in the conventional network, we introduce a {\it
reference} network by removing the PR link in each CR network. For
such reference networks, since there is no interference power
constraint nor interference from PR-Tx to CR terminals, it is easy
to show that CR $k$ with the largest $h_k$ among all the CRs should
be scheduled for transmission at each fading state due to D-TDMA.
Thus, we define the MDG of the reference network as
\begin{equation}\label{eq:MD gain}
\bar{\gamma}_0(K)=\kappa_0\mathbb{E}[\max_k h_k]
\end{equation}
where $\kappa_0=1/\mathbb{E}[h_k]$ is a constant, which is
independent of $k$.

From (\ref{eq:MD gain MAC})-(\ref{eq:MD gain}), it follows that the
MDGs for different CR networks all differ from the conventional MDG
for the reference network. We highlight their differences as
follows.
\begin{itemize}
\item For the C-MAC, it is observed from (\ref{eq:MD gain MAC}) that
the MDG is obtained by taking the maximum product between $h_k$ and
$\min(P,\Gamma/g_k)$ over all $k$'s, where the former also exists in
the conventional MDG given in (\ref{eq:MD gain}), while the latter
is a new term due to independent $g_k$'s over which CR transmitters
interfere with PR-Rx, thus named {\it T-MID}.

\item For the C-BC, it is observed from (\ref{eq:MD gain BC}) that
the MDG is obtained by taking the maximum product between $h_k$ and
$1/(1+Qe_k)$ over CRs, where the former term contributes to the
conventional MDG, while the latter term is a new source of diversity
due to independent $e_k$'s over which PR-Tx interferes with CR
receivers, thus named {\it R-MID}.

\item For the C-PAC, it follows from (\ref{eq:MD gain PAC}) that in
addition to the conventional MD, there are {\it combined} T-MID and
R-MID.
\end{itemize}

Next, we show the following theorem, which  says that the MDG of
each CR network for a given $K$ is lower-bounded by that of the
reference network, $\bar{\gamma}_0(K)$; and is upper-bounded by a
constant (greater than one) multiplication  of $\bar{\gamma}_0(K)$.

\begin{theorem}\label{theorem:MDG bounds} Under
the symmetric assumptions, we have $\bar{\gamma}_0(K)\leq
\bar{\gamma}_{\rm MAC}(K) \leq \alpha_{\rm MAC}\bar{\gamma}_0(K)$,
$\bar{\gamma}_0(K)\leq \bar{\gamma}_{\rm BC}(K) \leq \alpha_{\rm
BC}\bar{\gamma}_0(K)$, and $\bar{\gamma}_0(K)\leq \bar{\gamma}_{\rm
PAC}(K) \leq \alpha_{\rm PAC}\bar{\gamma}_0(K)$, where $\alpha_{\rm
MAC}=P/\mathbb{E}[\min(P,\Gamma/g_k)]$, $\alpha_{\rm
BC}=1/\mathbb{E}[1/(1+Qe_k)]$, and $\alpha_{\rm PAC}=\alpha_{\rm
MAC}\cdot\alpha_{\rm BC}$ are all constants.
\end{theorem}
\begin{proof}
We only show the proof for the C-PAC case, while similar proofs can
be obtained for the C-MAC and C-BC and are thus omitted here. First,
we consider the lower bound on $\bar{\gamma}_{\rm PAC}(K)$. By
denoting $k'$ as the user with the largest $h_k$ among all the CRs,
it follows that
\begin{align}
\bar{\gamma}_{\rm PAC}(K)&\overset{(a)}{\geq} \kappa_{\rm
PAC}\mathbb{E}[
h_{k'}\min(P,\Gamma/g_{k'})/(1+Qe_{k'})] \nonumber \\
&\overset{(b)}{=}\kappa_{\rm PAC}\mathbb{E}[
h_{k'}]\mathbb{E}[\min(P,\Gamma/g_{k'})]\mathbb{E}[1/(1+Qe_{k'})]
\nonumber \\ &\overset{(c)}{=} \bar{\gamma}_0(K) \nonumber
\end{align}
where $(a)$ is due to the fact that $k'$ is in general not the
optimal $k^*$ corresponding to the largest
$h_k\min(P,\Gamma/g_k)/(1+Qe_k)$ in (\ref{eq:MD gain PAC}); $(b)$ is
due to independence of channels $h_{k'}$, $g_{k'}$, and $e_{k'}$;
and $(c)$ is due to (\ref{eq:kappa PAC}) and (\ref{eq:MD gain}).
Next, we consider the upper bound on $\bar{\gamma}_{\rm PAC}(K)$.
The followings can be shown.
\begin{align}
\bar{\gamma}_{\rm PAC}(K)&\overset{(a)}{\leq} \kappa_{\rm PAC}
\mathbb{E}[\max_k
h_k \cdot \max_k\min(P,\Gamma/g_k)/(1+Qe_k)] \nonumber \\
&\overset{(b)}{=}\kappa_{\rm PAC} \mathbb{E}[\max_k
h_k]\mathbb{E}[\max_k\min(P,\Gamma/g_k)/(1+Qe_k)] \nonumber \\
&\overset{(c)}{\leq} \kappa_{\rm PAC} \mathbb{E}[\max_k h_k] P
\nonumber \\ &\overset{(d)}{=}
\bar{\gamma}_0(K)P/(\mathbb{E}[\min(P,\Gamma/g_k)]\mathbb{E}[1/(1+Qe_k)])\nonumber
\end{align}
where $(a)$ is due to the fact that the user with the largest $h_k$
is not necessarily the one with the largest
$\min(P,\Gamma/g_k)/(1+Qe_k)$; $(b)$ is due to independence of $h_k$
and $(g_k,e_k)$; $(c)$ is due to the fact that
$\min(P,\Gamma/g_k)/(1+Qe_k)\leq P, \forall k$; and $(d)$ is due to
(\ref{eq:kappa PAC}) and (\ref{eq:MD gain}). Using the definitions
of $\alpha_{\rm MAC}$ and $\alpha_{\rm BC}$ given in Theorem
\ref{theorem:MDG bounds}, it follows that $\bar{\gamma}_{\rm PAC}(K)
\leq \alpha_{\rm PAC}\bar{\gamma}_0(K)$, where $\alpha_{\rm
PAC}=\alpha_{\rm MAC}\cdot\alpha_{\rm BC}$.
\end{proof}

At last, we study the {\it ergodic throughput} of each D-TDMA based
CR network and its asymptotic growth order over $K$ as
$K\rightarrow\infty$. For the C-MAC, the ergodic throughput for a
given $K$ is defined as $C_{\rm
MAC}(K)=\mathbb{E}[\log_2(1+\gamma_{\rm MAC}(K))]$; similarly,
$C_{\rm BC}(K)$, $C_{\rm PAC}(K)$, and $C_0(K)$ are defined for the
C-BC, C-PAC, and reference network, respectively. According to the
extreme value theory \cite[Appendix A]{Sharif05}, it is known that
$C_0(K)$ behaves as $\frac{C_0(K)}{\log_2F(K)}\rightarrow 1$ as
$K\rightarrow\infty$, where $F(K)$ is given by the distribution of
$\max_{k}\mv{h}_k$ as $K\rightarrow\infty$ (e.g., for ``type i''
distribution of $\mv{h}_k$ with unit-variance, $F(K)=\log K$
\cite{Sharif05}). Then, from Theorem \ref{theorem:MDG bounds}, the
following corollary can be obtained (proof is omitted here due to
the space limitation).
\begin{corollary}\label{corollary:asymptotics}
Under the symmetric assumptions, as $K\rightarrow \infty$, we have
$\frac{C_{\rm MAC}(K)}{\log_2 F(K)}=1$,  $\frac{C_{\rm
BC}(K)}{\log_2 F(K)}=1$, and $\frac{C_{\rm PAC}(K)}{\log_2 F(K)}=1$.
\end{corollary}

Corollary \ref{corollary:asymptotics} says that the MID results in
the same ergodic-throughput asymptotic growth order over $K$ for the
CR networks  as that of the conventional MD for the reference
network, as $K\rightarrow\infty$, regardless of the fading
distribution.

\section{Numerical Results}

We assume that all the channels involved follow the standard
(unit-power) Rayleigh fading distribution. In addition, for a fair
comparison of different CR networks and the reference network, we
assume that $J=Q=P=\Gamma=1$. Fig. \ref{fig:rate low K} shows the
ergodic throughput for different networks with $K\leq 100$, after
normalizing it to the ergodic throughput with $K=1$ in order to
better examine the MDG. It is observed that the normalized ergodic
throughput for different CR networks is larger than that for the
reference network, thanks to the newly discover MID. It is also
observed that the combined T-MID and R-MID in the C-PAC result in
more substantial throughput gains than T-MID in the C-MAC or R-MID
in the C-BC. Fig. \ref{fig:rate high K} shows the ergodic throughput
(without normalization) of different networks versus $\log_2(\log
K)$ (Note that $F(K)=\log K$ in this case.) for very large values of
$K$ ranging from $10^3$ to $10^6$. It is observed that all these
networks have the same asymptotic growth order over $K$ for the
ergodic throughput, which is in accordance with Corollary
\ref{corollary:asymptotics}.

\begin{figure}
\centering{
 \epsfxsize=3.4in
    \leavevmode{\epsfbox{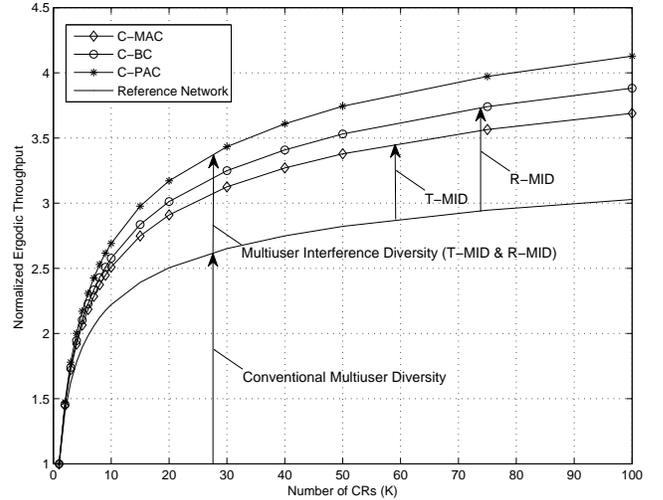}} }\vspace{-0.1in}
\caption{Normalized ergodic throughput for CR networks with $K\leq
100$.}\label{fig:rate low K}
\end{figure}

\begin{figure}
\centering{
 \epsfxsize=3.4in
    \leavevmode{\epsfbox{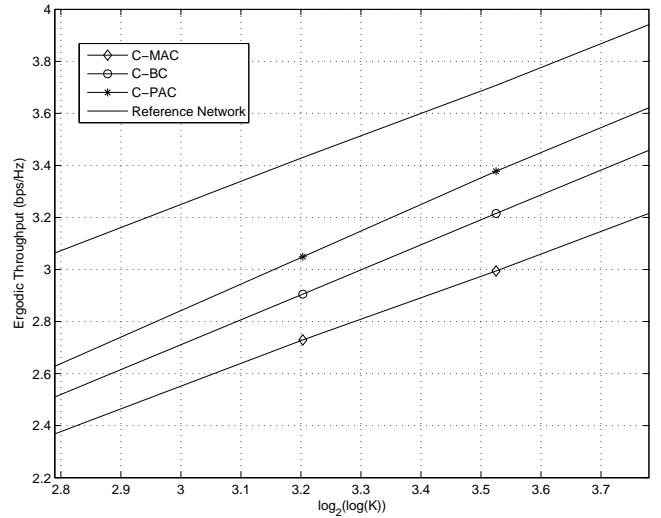}} }\vspace{-0.1in}
\caption{Asymptotic ergodic throughput for CR networks as
$K\rightarrow\infty$.}\label{fig:rate high K}
\end{figure}
\section{Conclusion}\label{sec:conclusion}

This letter quantified a new form of MID for SS-based CR networks by
exploiting the CR and PR mutual interference. The MDG and ergodic
throughput for opportunistic communications in different types of CR
networks were analyzed.

\end{document}